\documentclass{lmcs}


\newcommand{\diff}[3]{\operatorname{diff}_{#1}^{#2}\left(#3\right)}

\newcommand{\dif}[1]{\operatorname{diff}_{{#1}}}

\newcommand{\wrte}[3]{\operatorname{write}({#1}, {#2}, {#3})}

\usepackage{tikz-cd}
\usetikzlibrary{positioning}

\usepackage{stmaryrd}

\newcommand{\tc}{T_{\text{CAL}}}

\newcommand{\map}[1]{\operatorname{map}_{#1}}

\newcommand{\as}{\operatorname{ARRAY}}

\newcommand{\es}{\operatorname{ELEM}}

\newcommand{\is}{\operatorname{INDEX}}

\newcommand{\ar}[1]{\text{ar}(#1)}

\usepackage{cases}

\begin{document}

\title[Towards Interpolation in Unbounded Array Languages]{Towards Quantifier-Free Interpolation in Array Languages with Unbounded Data Specifications}

\thanks{The first author wishes to thank Dr. Tanja Schindler for various discussions regarding the material presented in \cite{raya_interpolating_2025}. Research supported by the Santander Postdoctoral Fellowship SN260020008RRC}	

\author[R.~Raya]{Rodrigo Raya\lmcsorcid{0000-0002-0866-9257}}[a]
\author[C.~Ringeissen]{Christophe Ringeissen\lmcsorcid{0000-0002-5937-6059}}[b]

\address{Technical University of Madrid, 28031 Madrid, Spain}	

\address{Universit\'e de Lorraine, CNRS, Inria, LORIA, F-54000 Nancy, France}

\begin{abstract}
We investigate quantifier-free interpolation properties for several fragments generalising the extensional theory of arrays. Our results include the (general) quantifier-free interpolation properties of combinatory array logic with iterated diffs and the uniform interpolation property of the simple flat array fragment. To our knowledge, these are the first positive quantifier-free interpolation results obtained for theories of arrays featuring expressive specifications over unbounded domains. 
\end{abstract}
	
\maketitle

\section{Introduction}
Software model checkers employ a variety of automated reasoning techniques to develop correct programs. One of the challenges in model checking software, as opposed to hardware specifications, is the need to handle infinite domains. In this setting, one usually needs to work with a suitable abstraction of the system, specified in a decidable logical fragment. Upon certifying the correctness of the abstraction, a variety of program synthesis techniques can be used to produce source code in a programming language of interest.

One application of interpolation algorithms is to accelerate the discovery of program invariants, which are essential to approximate the semantics of programming patterns such as loops or function calls. For instance, in CEGAR-based model checkers, interpolation algorithms are used to refine abstractions based on the analysis of spurious error traces \cite{ kapur_interpolation_2006}. When these traces and their interpolants are given by quantifier-free formulas, model checkers can make use of fast decision procedures implemented in Satisfiability Modulo Theories (SMT) solvers. Among the most expressive theories supported by these solvers are data structures theories. Unfortunately, not all data structure theories admit quantifier-free interpolation \cite{hoenicke_interpolation_2019, ghilardi_interpolation_2023}, which limits the class of systems for which safety properties can be shown with the help of interpolation. 

Another application of interpolation is program synthesis under the assumption that certain components of the system are known \cite{kuncak_interpolation_2013}.

This paper addresses the property of (general and uniform) quantifier-free interpolation in data structure theories that generalise the extensionality axiom \cite{mccarthy_towards_1993} to functions and relation symbols. This extension of the quantifier-free theory of arrays is particularly useful in the verification of parametrised systems or distributed protocols \cite{dragoi_logic-based_2014, dragoi_psync_2016, gleissenthall_cardinalities_2016, alberti_cardinality_2017, damian_communication-closed_2019, konnov_tla_2019, ghilardi_higher-order_2020}, where one is often interested in guaranteeing that a certain property holds at every component of the system. It turns out that supporting such generalisation incurs only a mild overhead with respect to the extensional fragment \cite{de_moura_generalized_2009}. Moreover, recent work shows that the theory can be combined with cardinality constraints, preserving decidability \cite{ghilardi_higher-order_2020, raya_vmcai_2022}. Further extensions supporting regular relations and aggregation constraints are described in \cite{raya_succinct_2024}. 

In \cite{raya_interpolating_2025}, we developed the observation that the decision procedures in \cite{raya_vmcai_2022, raya_succinct_2024} rely on a kind of reduction to the index and element theory of the data structure, analogous to the reduction approach found in the method of interpolation of \cite{kapur_interpolation_2006}. We used this observation to show weak quantifier-free interpolation properties of combinatory array logic \cite{de_moura_generalized_2009}, the theory of sets with cardinalities \cite{zarba_combining_2005} and the simple flat array fragment  \cite{alberti_cardinality_2017}. We chose these as representatives of a class of quantifier-free theories, which we call collectively algebraic theories of arrays, because they extend the quantifier-free theory of power structures \cite{mostowski_direct_1952, feferman_first_1959} with read and write operations. 

In this paper, we use the methods of \cite{bruttomesso_quantifier-free_2012} to show that these theories admit full quantifier-free interpolation properties (i.e. where the interpolant only contains shared symbols). In this way, we extend the applicability of interpolation algorithms to array theories that include rich and unbounded properties beyond extensional equality. 

\vspace{.25em}

\noindent \textbf{Related work.} Our results close the open question in \cite{ghilardi_interpolation_2023} about the quantifier-free interpolation property of the simple flat array fragment and relate to a number of problems in the specification and verification of fault-tolerant distributed protocols \cite{dragoi_need_2015} and invariant inference \cite{beyer_invariant_2007}.

Early work on quantifier-free interpolation focused on theories of data structures and linked the property to quantifier elimination \cite{kapur_interpolation_2006}. \cite{bruttomesso_quantifier-free_2012} showed how to recover quantifier-free interpolation for theories of arrays by using a relaxed notion of interpolation that only required the existence of quantifier-free interpolants for quantifier-free formulas and extended the theory of arrays with diff functional terms. \cite{bruttomesso_quantifier-free_2014} introduced the notion of general quantifier-free interpolants (similar notions exist in the unification theory \cite{baader_unification_1996}). A parallel line of work develops interpolation algorithms for local theory extensions \cite{sofronie-stokkermans_hierarchic_2005,sofronie-stokkermans_interpolation_2006, sofronie-stokkermans_interpolation_2008}. \cite{totla_complete_2013, totla_complete_2016} introduce the notion of $W$-separable theory extension and show that ground quantifier-free interpolation is available for these theories. \cite{sofronie-stokkermans_interpolation_2016} extends this work showing certain general ground quantifier-free interpolation results, while leaving open the question of uniform interpolation. It is known that BAPA can be seen as a local theory extension \cite{jacobs_hierarchic_2009}, but at present such results are not available for combinatory array logic or the simple flat array fragment. Moreover, note that there seems to be some evidence that the abstract interpolation procedures in \cite{totla_complete_2013} fail for some extensionality lemmas \cite{hoenicke_efficient_2018}. Finally, some work exploring the relationship between Feferman-Vaught type theorems and interpolation properties is available \cite{caicedo_failure_1985, makowsky_compactness_2017}, but none of them target quantifier-free fragments. 

In recent papers \cite{castellanos_joo_axdinterpolator_2021, ghilardi_interpolation_2023}, Castellanos et alii considered interpolation algorithms for a theory of contiguous arrays extended with length and maxdiff operators. They proved that this theory enjoys quantifier-free interpolation but lacks the general interpolation property. By adding constants to the language, they obtain a theory that has these properties. The theories we study in this paper are different, since they do not assume a linear order in the index theory (cf. \cite[Definition~16]{ghilardi_interpolation_2023}). In spite of this, note that there exist decision procedures for more expressive theories of indices \cite{alberti_cardinality_2017, raya_succinct_2024}, which can be obtained using adaptations of the Feferman-Vaught theorem to the quantifier-free fragment. A second difference is that we also study the property of uniform interpolation, which is not present in \cite{ghilardi_interpolation_2023}.

There are practical motivations behind the family of logical theories that we study.  \cite{dragoi_need_2015} argues that the field of fault-tolerant distributed computation needs high-level programming language support that is amenable to automated verification and testing. Further work has addressed this challenge \cite{dragoi_logic-based_2014, dragoi_psync_2016, gleissenthall_cardinalities_2016, alberti_cardinality_2017, damian_communication-closed_2019}, but the verification logics were invariably extensions of interpreted BAPA. A second relevant application of the theories studied is in invariant inference. There is a long line of work in the inference of invariants for theories admitting quantifier elimination \cite{kapur_quantifier-elimination_2006}. Partial versions of invariant inference of data structures occur in \cite{beyer_invariant_2007, gleissenthall_cardinalities_2016, konnov_tla_2019, ghilardi_higher-order_2020}. In contrast to invariant inference procedures, interpolation algorithms construct invariants iteratively, which we hope will make the problem more tractable than as is reported in \cite{beyer_invariant_2007}.

\vspace{.25em}

\noindent \textbf{Contributions.} We contribute to understanding the quantifier-free interpolation properties of array theories. In particular, we
\begin{itemize}
\item identify a new kind of diff functional (diff$_R$ for any relation $R$) term enabling quantifier-free interpolation properties in a family of array theories from the literature.
\item show the existence of quantifier-free interpolants for extensions of the theory of power structures containing componentwise operations (e.g. combinatory array logic).

\item provide examples of non-uniform interpolating theories of arrays and uniform interpolating theories using elimination of existential quantifiers (e.g. the simple flat array fragment). 
\end{itemize}



\vspace{.25em}

\noindent \textbf{Paper organisation.} Section~\ref{section:prelim} introduces preliminaries in first-order logic. Section~\ref{section:arrays} describes the array theories that we analyse. Section~\ref{section:motivation} gives a simple example of the existence of quantifier-free interpolants in combinatory array logic.   Section~\ref{section:model} gives background results on the relation between model-theoretic properties and interpolation. Section~\ref{section:qfinterp} shows quantifier-free interpolation properties for combinatory array logic. Section~\ref{section:gqfinterp} extends the previous results to general quantifier-free interpolation. Section~\ref{section:unifinterp} discusses the uniform interpolation property, obtaining mixed results (cardinality constraints are needed to have uniform interpolation). Section~\ref{section:conclusion} concludes the paper.

\section{Preliminaries}
\label{section:prelim}
We assume the usual syntactic (e.g., signature, variable, term, atom, literal, formula, and sentence) and semantic (e.g., structure, truth, satisfiability, and validity) notions of first-order logic. The equality symbol $=$ is included in all signatures considered below. 

\begin{defi}
A theory $T$ is a pair ($\Sigma, A x_{T}$), where $\Sigma$ is a signature and $A x_{T}$ is a set of $\Sigma$-sentences, called the axioms of $T$. The models of $T$ are the $\Sigma$-structures in which all sentences from $A x_{T}$ are true. 
\end{defi}

\begin{defi}
A universal (resp. existential) sentence is obtained by prefixing a string of universal (resp. existential) quantifiers to a quantifier-free formula. A theory $T$ is universal iff $A x_{T}$ consists of universal sentences. 
\end{defi}

\begin{defi}
A $\Sigma$-formula $\phi$ is $T$-satisfiable if there exists a model $\mathcal{M}$ of $T$ such that $\phi$ is true in $\mathcal{M}$ under a suitable assignment $a$ to the free variables of $\phi$ (in symbols, ($\mathcal{M}, \mathrm{a}) \models \phi$); it is $T$-valid (in symbols, $T \vdash \varphi$ ) if its negation is $T$-unsatisfiable or, equivalently, if and only if $\varphi$ is provable from the axioms of $T$ in a complete calculus for first-order logic. A formula $\varphi_{1} T$-entails a formula $\varphi_{2}$ if $\varphi_{1} \rightarrow \varphi_{2}$ is $T$-valid; the notation used for such $T$-entailment is $\varphi_{1} \vdash_{T} \varphi_{2}$ or simply $\varphi_{1} \vdash \varphi_{2}$, if $T$ is clear from the context. 
\end{defi}

\begin{defi}
If $\Sigma$ is a signature, we use the notation $\mathcal{M}=(M, \mathcal{I})$ for a $\Sigma$-structure, meaning that $M$ is the support of $\mathcal{M}$ and $\mathcal{I}$ is the related interpretation function for $\Sigma$-symbols.
\end{defi}

\begin{defi}
A $\Sigma$-embedding (or, simply, an embedding) between two $\Sigma$-structures $\mathcal{M}=(M, \mathcal{I})$ and $\mathcal{N}=(N, \mathcal{J})$ is any mapping $\mu: M \longrightarrow N$ among the corresponding support sets satisfying the following three conditions: 
\begin{enumerate}
\item $\mu$ is a (sort-preserving) injective function; 

\item $\mu$ is an algebraic homomorphism, that is for every $n$-ary function symbol $f$ and for every $a_{1}, \ldots, a_{n} \in M$, we have $f^{\mathcal{N}}\left(\mu\left(a_{1}\right), \ldots, \mu\left(a_{n}\right)\right)=\mu\left(f^{\mathcal{M}}\left(a_{1}, \ldots, a_{n}\right)\right)$; 

\item $\mu$ preserves and reflects interpreted predicates, i.e. for every $n$-ary predicate symbol $P$, we have $\left(a_{1}, \ldots, a_{n}\right) \in P^{\mathcal{M}}$ iff $\left(\mu\left(a_{1}\right), \ldots, \mu\left(a_{n}\right)\right) \in P^{\mathcal{N}}$. 
\end{enumerate}
\end{defi}

\begin{prop}
Every embedding can be factored in an isomorphism and an inclusion. 
\end{prop}
\begin{proof}
Every embedding $\mu$ gives an isomorphism $g$ with its image, which is a structure closed under formation of terms. Since the image is included in the co-domain, it follows that $\mu = i \circ g$ where $i$ is the inclusion mapping.    
\end{proof}

\begin{defi}
If $M \subseteq N$ and the embedding $\mu: \mathcal{M} \longrightarrow \mathcal{N}$ is just the identity inclusion, we say that $\mathcal{M}$ is a substructure of $\mathcal{N}$ or that $\mathcal{N}$ is an superstructure of $\mathcal{M}$. 
\end{defi}

Embeddings preserve
existential and universal formulas in the following sense.

\begin{prop}
Suppose $\eta : \mathcal{M} \to \mathcal{N}$ is an embedding, $a \in M$, and $\phi(v)$ is an existential formula and $\psi(v)$ is a universal formula. If $\mathcal{M} \models \phi(a)$, then $\mathcal{N} \models \phi(\eta(a))$.
 If $\mathcal{N} \models \psi(\eta(a))$, then $\mathcal{M} \models \psi(a)$.
\end{prop}
\begin{proof}
It follows from Theorem~2.5 and Corollary~2.6 in \cite{david_marker_invitation_2024}.
\end{proof}

As a consequence, 

\begin{cor} \label{cor:preserv}
The truth of a universal (resp. existential) sentences is preserved through substructures (resp. through superstructures).
\end{cor}

We conclude with the definition of quantifier-free interpolants.

\begin{defi}
A theory $T$ is said to admit quantifier-free interpolation (or, equivalently, to have quantifier-free interpolants) if and only if for every pair of quantifier free formulae $\phi, \psi$ such that $\psi \wedge \phi$ is not $T$ satisfiable, there exists a quantifier free formula $\theta$, called an interpolant, such that: (i) $\psi$ $T$-entails $\theta$; (ii) $\theta \wedge \phi$ is not $T$-satisfiable: (iii) only variables occurring both in $\psi$ and in $\phi$ occur in $\theta$.
\end{defi}

\section{Theories of Arrays}
\label{section:arrays}
Specification languages for arrays have existed since the early days of program verification \cite{king_program_1969}, borrowing ideas of McCarthy \cite{mccarthy_towards_1993}. The theory of arrays with an extensionality axiom was studied in \cite{stump_decision_2001}. Bruttomesso et alii \cite{bruttomesso_quantifier-free_2012} studied quantifier-free interpolation for this extensional array theory. All of these theories provide syntax to specify reads and writes, but are unable to specify richer properties that hold globally in the array.

In this paper, we investigate array fragments that specify properties that hold globally in the arrays. Our starting point for such array theories is the combinatory array logic fragment \cite{de_moura_generalized_2009}. We also investigate the richer flat array fragment \cite{alberti_cardinality_2017} since it can be used to encode regular properties on the index set and trees \cite{raya_succinct_2024}. In describing these fragments, we assume that the component specifications are given with uninterpreted functions and relations. This is because it is not yet clear under which conditions one can derive interpolation algorithms for specific data theories, a question we plan to answer following the methodology in \cite{bruttomesso_quantifier-free_2014}. At the same time, we plan to investigate whether the stably infinite assumption used in \cite{bruttomesso_quantifier-free_2014} can be lifted. In our research on decision procedures \cite{raya_polite_2025}, we found that the stably infinite assumption is not necessary for combining decision procedures in the context of the studied array theories.

\begin{figure}
\centering
\begin{align*}
F & ::= F_1 \land F_2 \, | \, F_1 \lor F_2 \, | \, \lnot F \, | \, \operatorname{map}_R(\overline{A}) \, | \, A[i] = A[j] \, A[i] = e \\
A & ::= a \, | \, \operatorname{write}(A,i,E) \, | \, K(e) \, | \, \operatorname{map}_f(\overline{A}) 
\end{align*}
\caption{$T_{\text{CAL}}$'s syntax.}
\label{fig:cal-syntax}
\end{figure}

Combinatory array logic \cite{de_moura_generalized_2009}, $\tc$, uses a quantifier-free many-sorted language, comprising an index sort, an element sort, and an array sort. In Figure~\ref{fig:cal-syntax}, we use instances of letter $i$ to refer to variables of the index sort, letter $e$ to refer to variables of sort element and letter $a$ to refer to a variable of the array sort. 

For the purpose of this paper, we will fix a collection of function symbols $f_1,\ldots,f_p$ and relation symbols $R_1,\ldots,R_k$, with the aim of describing a first-order theory. In this way, we will avoid having to extend the proofs relating  amalgamation and quantifier-free interpolation. 

Combinatory array logic provides the following functions and operators over the parametrised signature. Function symbol $\operatorname{write}$ takes as input an array, an index and an element and returns an array. Function read, $\_[\_]$, takes an array and an index and returns an element. Function $K$ takes as input an element and returns an array which is constantly equal to that element. Function $\operatorname{map}_f(\overline{A})$ takes as input a tuple of arrays and returns an array where the $i$-th component is the result of applying function $f$ to the $i$-th component of the input arrays. Relation $\operatorname{map}_R(\overline{A})$ takes as input a tuple of arrays and returns true if and only if relation $R$ holds for each tuple of elements in the $i$-th component. Formally, the introduced symbols satisfy the following axioms. \begin{align*}
& \forall a, i, e . \operatorname{write}(a, i, e)[i] = e \\
& \forall a, i, j, e . i = j \vee \operatorname{write}(a, i, e)[j] = a[j] \\
& \forall e, i . K(e)[i] = e \\
& \forall a_{1}, \ldots, a_{k}, i.\operatorname{map}_{f}\left(a_{1}, \ldots, a_{k}\right)[i] = f\left(a_{1}[i], \ldots, a_{k}[i]\right) \\
& \forall a_{1}, \ldots, a_{k}. \operatorname{map}_R(a_1,\ldots,a_k) \leftrightarrow \forall i. R(a_1[i],\ldots,a_k[i])
\end{align*}

For the purpose of illustrating the quantifier-free interpolation properties, we will consider a purely relational signature (i.e. without function symbols) and thus, we will work with the following simplified axiomatisation:
\begin{align*}
\forall a, i, e . & write(a, i, e)[i]=e  \tag{3.1} \\
\forall a, i, j, e . & i \neq j \to write(a, i, e)[j] = a[j]  \tag{3.2} \\
\forall a_1,\ldots,a_n,i. &  \text{map}_R(a_1,\ldots,a_n) \rightarrow  R(a_1[i],\ldots,a_n[i]) \tag{3.3} \label{ax:map} \\
\forall a_1,\ldots,a_n. &  \text{map}_R(a_1,\ldots,a_n) \lor  \lnot R(a_1[\diff{R}{1}{a_1,\ldots,a_n}],\ldots,a_n[\diff{R}{1}{a_1,\ldots,a_n}]) \tag{3.4} \label{ax:diff} \\
\forall a_1,\ldots,a_n, i. &  \begin{aligned}[t]
& (\diff{R}{j-1}{a_1,\ldots,a_n} \neq \diff{R}{j}{a_1,\ldots,a_n} \land \\
& \lnot R(a_1[\diff{R}{j}{a_1,\ldots,a_n})],\ldots,a_n[\diff{R}{j}{a_1,\ldots,a_n})]) \lor \\ & (\diff{R}{j-1}{a_1,\ldots,a_n} = \diff{R}{j}{a_1,\ldots,a_n} \land  \\ & (i \neq \diff{R}{1}{a_1,\ldots,a_n} \land \ldots \land i \neq \diff{R}{j-1}{a_1,\ldots,a_n} \to R(a_1[i],\ldots,a_n[i]))) \label{ax:iterated}
\end{aligned} 
\tag{3.5}
\end{align*}
where the last axiom is formulated for each natural number $j > 0$.

Note that Axioms~\ref{ax:map} and \ref{ax:diff} originate from a transformation of the original axiom of \cite{de_moura_generalized_2009}
\begin{align*}
\forall a_1,\ldots,a_n.   \map{R}(a_1,\ldots,a_n) \leftrightarrow  \forall i. R(a_1[i],\ldots,a_n[i])
\end{align*}
in order to ensure that the theory is universal. The left to right implication is kept and the right to left is Skolemised using the symbol $\dif{R}$. Moreover, we have iterated $\dif{R}$ operators. An intuition on why these are necessary is provided in the next section.

The simple flat array fragment \cite{alberti_cardinality_2017, ghilardi_higher-order_2020, raya_vmcai_2022}, $T_{\text{F}}$, uses a quantifier-free many-sorted language, comprising a sort for indices, a sort for elements, a sort for integers, a sort for sets of indices, and a sort for arrays.  In Figure~\ref{fig:alberti-syntax}, we use (indexed) instances of letter $i$ to refer to variables of the index sort, letter $x$ to refer to variables of sort set, letter $k$ to refer to a variable of the integer sort and letter $a$ to refer to a variable of the array sort. The logic has as atoms set comprehensions of the form $\{ i \mid \varphi(\overline{a}[i]) \}$ where $\overline{a}$ is a set of array variables, $\overline{a}[i]$ denotes the set of elements in the $i$-th component of the array variables and $\varphi$ is a formula in the theory of elements (here assumed to be the theory of equality with uninterpreted function symbols). These set comprehensions are combined with Boolean algebra expressions and cardinality constraints. The semantics of these symbols is the standard.

\begin{figure}
\centering
\begin{align*}
F & ::= A \, | \, F_1 \land F_2 \, | \, F_1 \lor F_2 \, | \, \lnot F \\
A & ::=
i_1 = i_2 \, | \, i \in B \, | \,  B_1 \subseteq B_2 \, | \, T_1 \le T_2 \\
B & ::= x \, | \, \emptyset \, | \, B_1 \cup B_2 \, | \, B_1 \cap B_2 \, | \, B^c \, | \, \{ i \mid \varphi(\overline{a}[i]) \} \\
T & ::= k \, | \, K \, | \, T_1 + T_2 \, | \, K \cdot T \, |  \, |B| \\
K & ::= \ldots \, | \, -2 \, | \, -1 \, | \, 0 \, | \, 1 \, | \, 2 \, | \, \ldots
\end{align*}
\caption{$T_{\text{F}}$'s syntax}
\label{fig:alberti-syntax}
\end{figure}

\section{Motivating Example}
\label{section:motivation}
With the intent of providing the reader with a concrete non-trivial example of quantifier-free interpolant and demonstrating the use of the $\dif{R}$ operators let us consider the following two formulas. The first formula $\phi$ is given by  
\[
\map{R}(a) \land b = \wrte{a}{i}{v}
\]
The second formula $\psi$ is given by 
\[
\lnot R(b[j]) \land \lnot R(b[k]) \land j \neq k
\]

It is clear that $\phi \land \psi$ is unsatisfiable in $\tc$.

We wish to demonstrate that there exists an interpolant in the language of combinatory array logic with $\dif{R}$ operators. An interpolant formula can only use the common variable $b$ and must express the fact that $b$ contains zero or one occurrence of an element satisfying the relation $\lnot R$.

Note that an interpolation algorithm based on grounding the map relations over the (global) set of index variables may introduce $i,j,k$ as common variables, even though they are local to each $\phi$ and $\psi$. Thus, this approach is not possible if we want to produce strong interpolants.

At first, it may seem that the following formula works as an interpolant.
\[
map_R(b) \lor (\lnot R(b[\diff{R}{}{b}]) \land \map{R}(\wrte{b}{\diff{R}{}{b}}{b[\diff{\lnot R}{}{b}]})) 
\]
However, this formula does not work when the model of $\phi$ is an array of a single position and content $v$ satisfying $\lnot R$. The reason is that $b[\diff{\lnot R}{}{b}]$ will return $v$ and therefore neither the first nor the second disjunct hold. 

Instead, we resort to the notion of iterated diff functions, first described in \cite{ghilardi_interpolation_2023}. Then, 
\[
map_R(b) \lor (\lnot R(b[\diff{R}{1}{b}]) \land \diff{R}{1}{b} = \diff{R}{2}{b} ) 
\]
works as a quantifier-free interpolant where the iterated diff gives the next index that has an element not satisfying $R$ if such an index exists and otherwise returns one of the previous diff indices. In our case, since we have that $\diff{R}{1}{b} = \diff{R}{2}{b}$ holds, we know that there are no more indices such that $\lnot R$ holds. Thus, the array has either zero or one positions where $\lnot R$ holds, as required.

\section{Amalgamation implies quantifier-free interpolation}
\label{section:model}

We now review the relationship between amalgamation and quantifier-free interpolation. In \cite{bruttomesso_quantifier-free_2012}, the existence of quantifier-free interpolants is shown to be implied by the property of amalgamation of the models of the theory.

\begin{defi}
\label{def:amalgamation}
A theory $T$ is said to have the amalgamation property if and only if whenever we are given embeddings
$$
\mu_{0}: \mathcal{N} \longrightarrow \mathcal{M}_{0}, \quad \mu_{1}: \mathcal{N} \longrightarrow \mathcal{M}_{1}
$$
among the models $\mathcal{N}, \mathcal{M}_{0}, \mathcal{M}_{1}$ of $T$, then there exists a further model $\mathcal{M}$ of $T$ endowed with embeddings
$$
\nu_{0}: \mathcal{M}_{0} \longrightarrow \mathcal{M}, \quad \nu_{1}: \mathcal{M}_{1} \longrightarrow \mathcal{M}
$$
such that $\nu_{0} \circ \mu_{0}=\nu_{1} \circ \mu_{1}$. 
\end{defi}

\begin{prop} \label{prop:inclusion}
Up to isomorphism, we can limit ourselves in the above definition to the case in which $\mu_0, \mu_1$ are inclusions, i.e. to the case in which $\mathcal{N}$ is just a substructure of both $\mathcal{M}_0, \mathcal{M}_1$.
\end{prop}
\begin{proof}
It is sufficient to see that for each embedding $\mu: \mathcal{N} \to \mathcal{M}$, we can build a structure $\mathcal{M}'$ isomorphic to $\mathcal{M}$ whose carrier set contains the elements of $N$. 

The carrier set of $\mathcal{M}'$ consists of the elements of $N$ plus the elements of $M$ minus those in the image of $\mu$. Then the isomorphism between $\mathcal{M}'$ and $\mathcal{M}$ sends the elements of $N$ to the corresponding elements in $Img(\mu)$ and the elements outside $N$ to the same in the carrier of $\mathcal{M}$. Bijectivity follows from the injectivity of $\mu$. The functions and relations on $\mathcal{N}$ are preserved by the definition of $\mu$. Any other function or relation is defined in $\mathcal{M}'$ as in $\mathcal{M}$.

Finally, if we had the amalgamation properties for structures $\mathcal{N}, \mathcal{M}_0'$ and $\mathcal{M}_1'$ then the amalgamation property for $\mathcal{M}_0$ and $\mathcal{M}_1$ follows from the commutativity of the following diagram.

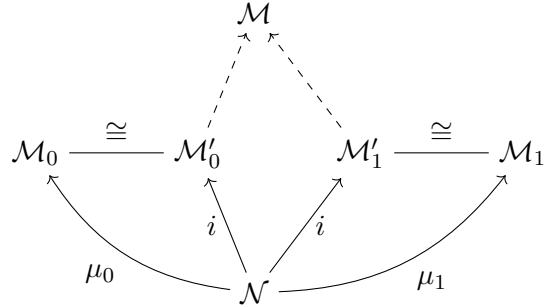
\begin{center}
\begin{tikzpicture}[node distance=1.25cm]

\node (M1) {$\mathcal{M}_0$};
\node[right=of M1] (A) {$\mathcal{M}_0'$};
\node[right=of A] (B) {$\mathcal{M}_1'$};
\node[right=of B] (M2) {$\mathcal{M}_1$};

\node[below=of A, xshift=0.75cm] (N) {$\mathcal{N}$};

\node[above=of A, xshift=0.75cm] (T) {$\mathcal{M}$};

\draw[->] (N) -- node[left] {$i$} (A);
\draw[->] (N) -- node[right] {$i$} (B);

\draw[->] (N) to[bend left=25]
    node[below left] {$\mu_0$} (M1);

\draw[->] (N) to[bend right=25]
    node[below right] {$\mu_1$} (M2);

\draw[dashed,->] (A) -- (T);
\draw[dashed,->] (B) -- (T);

\draw (M1) -- node[above] {$\cong$} (A);
\draw (M2) -- node[above] {$\cong$} (B);


\end{tikzpicture}%
\end{center}
\end{proof}



In the following, we denote by $\tc$ the theory of combinatory array logic with diff operators. 

\begin{thm} \label{thm:amalg-interp}
If $\tc$ has the amalgamation property then $\tc$ admits quantifier-free interpolants.
\end{thm}
\begin{proof}
This follows from \cite[Theorem~2.3]{bruttomesso_quantifier-free_2012}.
\end{proof}

\begin{exa}
Revisiting the motivating example of Section~\ref{section:motivation}, observe that combinatory array logic with diffs (but no iterated diffs) would not have the amalgamation property. This can be seen taking as the base of the amalgamation diagram a structure containing a single array $b$ of length one and satisfying a given relation $\lnot R$. This structure is embedded into a structure $\mathcal{M}_0$ containing the same array $b$ and a structure $\mathcal{M}_1$ containing an array $b$ that now has two components ($i$ and $j$) satisfying $\lnot R$. The interpretation of $\tc$'s symbols is the standard one. Then it is clear that $\mathcal{M}_0 \models \phi$ and $\mathcal{M}_1 \models \psi$. If the amalgamation property was true, there would be a model $\mathcal{M}$ of $\tc$ satisfying both $\phi$ and $\psi$. But this is not possible since $\phi \land \psi$ is $T$-inconsistent.

Note how the above scenario is not possible with iterated diffs. Both $\dif{R}^1$ and $\dif{R}^2$ would have to be mapped to the same index in $\mathcal{M}_1$, but this would not satisfy the axiomatisation of $\dif{R}^2$.
\end{exa}

\section{ Quantifier-free Interpolation}
\label{section:qfinterp}
To establish the amalgamation property for combinatory array logic, we first do a careful study of the shape of its models. 

\begin{defi}
A model of $\tc$ is \textit{standard} when the sort ARRAY is interpreted as the set of all functions from the sort INDEX to the sort ELEM, arrays are interpreted as functions, $\cdot[\cdot]$ as function application, $\wrte{\cdot}{\cdot}{\cdot}$ as the point-wise update operation, $\map{R}$ as the pointwise relation $R$, diff$_R^1$ as a function mapping a tuple of functions to an index where $\lnot R$ holds and to an arbitrary index if no such index exists and diff$_R^i$ for $i > 0$ is a function mapping a tuple of functions to a new index where $\lnot R$ holds or equal to the previous index in the iterated sequence in which case the rest of the indices of the array hold values where $\lnot R$ is true.

A  \textit{functional} model is a model $\mathcal{M}$ in which ARRAY${ }^{\mathcal{M}}$ is a subset of the set of functions from INDEX${ }^{\mathcal{M}}$ to ELEM${ }^{\mathcal{M}}$ and in which $\cdot[\cdot]^{\mathcal{M}}, \wrte{\cdot}{\cdot}{\cdot}^{\mathcal{M}}$, $\map{R}$ and diff$_R^i$ for $i \ge 0$ have their standard meaning.
\end{defi}
 
\begin{prop} \label{prop:embedding}
Every model of $\tc$ embeds into a standard one. 
\end{prop}
\begin{proof}
Let $\mathcal{M}$ be an arbitrary model of $\tc$. We  build a standard model $\operatorname{std}(\mathcal{M})$ with  
\[
\operatorname{INDEX}^{\text{std}(\mathcal{M})}=\operatorname{INDEX}^{\mathcal{M}} \text{ and } \operatorname{ELEM}^{\text {std}(\mathcal{M})}=\operatorname{ELEM}^{\mathcal{M}}
\]

We embed $\mathcal{M}$ into $\operatorname{std}(\mathcal{M})$ mapping every $a \in \operatorname{ARRAY}^{\mathcal{M}}$ to the function $\mu(a)$ such that 
\[
\mu(a)(i) := a[i]^{\mathcal{M}}
\]

We define the element sort relations in $\operatorname{std}(\mathcal{M})$ as in $\mathcal{M}$.


We also define diff$_R^{\operatorname{std}(\mathcal{M})}(a_1,\ldots,a_n) := $ diff$_R^{\mathcal{M}}(\mu^{-1}(a_1),\ldots,\mu^{-1}(a_n))$ if all inverses are defined and otherwise according to the standard semantics. We proceed analogously with the iterated diff functions.

$\mu$ is injective since $\mu(a) = \mu(b)$ implies that $a[i]^{\mathcal{M}} = b[i]^{\mathcal{M}}$ for every index $i$. Thus, from the extensionality axiom it follows that $a = b$. 

Next, let us show that $\map{R}$ relations are preserved by $\mu$. 

If $\overline{a} \in \map{R}^{\mathcal{M}}$ and $\mu(\overline{a}) \notin \map{R}^{\text{std}(\mathcal{M})}$ then $\lnot R(\mu(\overline{a})(\diff{R}{}{\overline{a}}))$ which holds if and only if $\lnot R(\overline{a}[\diff{R}{}{\overline{a}}])$ contradicting the axiomatisation of $\map{R}$. 

Conversely, if $\overline{a} \notin \map{R}^{\mathcal{M}}$ and $\mu(\overline{a}) \in \map{R}^{\text{std}(\mathcal{M})}$ then on the one hand we have that $\lnot R(\overline{a}[\diff{R}{}{\overline{a}}])$ and on the other hand we have that $R(\mu(\overline{a})(\diff{R}{}{\overline{a}}))$, which holds if and only if $R(\overline{a}[\diff{R}{}{\overline{a}}])$. This leads to a contradiction.

Finally, $\mu(\diff{R}{\mathcal{M}}{a_1,\ldots,a_n}) = \diff{R}{\mathcal{M}}{a_1,\ldots,a_n} = \diff{R}{std(\mathcal{M})}{\mu(a_1),\ldots,\mu(a_n)}$ follows from the definition. The proof is analogous for iterated diffs.
\end{proof}

\begin{cor} \label{cor:functional}
Every model of $\tc$ is isomorphic to a functional one.
\end{cor}
\begin{proof}
Every model is isomorphic to its image via the embedding of Proposition~\ref{prop:embedding}.
\end{proof}

\begin{cor}
To test the satisfiability of $\tc$ constraints, we can consider only standard models.
\end{cor}
\begin{proof}
This follows from the fact that embeddings preserve existential sentences (see Proposition~\ref{cor:preserv}). 
\end{proof}

\begin{defi}
Two elements $a, b$ of $\operatorname{ARRAY}^{\mathcal{M}}$ in a model $\mathcal{M}$ of $\tc$ are cardinality dependent, written $\mathcal{M} \models|a-b|<\omega$, if and only if $\left\{i \in \operatorname{INDEX}^{\mathcal{M}} \mid \mathcal{M} \models a[i] \neq\right. b[i] \}$ is finite. 
\end{defi}

\begin{prop}
Cardinality dependence is an equivalence relation.
\end{prop}
\begin{proof}
Reflexivity, symmetry and transitivity are easy to establish.
\end{proof}

The following lemma shows that cardinality dependence is preserved under embeddings.

\begin{lem} \label{lem:cardinality}
Let $\mathcal{N}, \mathcal{M}$ be models of $\tc$ and let $\mu:\mathcal{N} \to \mathcal{M}$ be an embedding that restricts to an inclusion $\is^{\mathcal{N}} \subseteq \is^{\mathcal{M}}$, $\es^{\mathcal{N}} \subseteq \es^{\mathcal{M}}$. For every $a, b \in \as^{\mathcal{M}}$, we have that
$$
\mathcal{N} \models|a-b|<\omega \quad \text { iff } \quad \mathcal{M} \models|\mu(a)-\mu(b)|<\omega .
$$
\end{lem}
\begin{proof}
If $\mathcal{N} \models|a-b|<\omega$ then $\mathcal{N} \models a=\wrte{b}{I}{E}$, where $I \equiv i_{1}, \ldots, i_{n}$ is a list of terms of sort INDEX, $E \equiv e_{1}, \ldots, e_{n}$ is a list of terms of sort ELEM, and $\wrte{b}{I}{E}$ abbreviates the term $\wrte{\cdots \wrte{a}{i_{1}}{e_{1}} \cdots}{i_{n}}{e_{n}}$. Thus, also $\mathcal{N} \models \mu(a)= \wrte{\mu(b)}{I}{E})$. Thus, we have that $\mathcal{N} \models |\mu(a)-\mu(b)| < \omega$.

Conversely, if $\mathcal{N} \not \models|a-b|<\omega$ then there are infinitely many $i \in \operatorname{INDEX}^{\mathcal{N}}$ such that $a^{\mathcal{N}}[i]^{\mathcal{N}} \neq b^{\mathcal{N}}[i]^{\mathcal{N}}$. Since $\mu$ is injective, there are also infinitely many $i \in \operatorname{INDEX}^{\mathcal{M}}$ such that $\mu(a)^{\mathcal{M}}[i]^{\mathcal{M}} \neq \mu(b)^{\mathcal{M}}[i]^{\mathcal{M}}$, i.e. $\mathcal{M} \not \models|\mu(a)-\mu(b)|<\omega$.
\end{proof}

\begin{lem} \label{lem:extension}
Let $\mathcal{N}, \mathcal{M}$ be models of $\tc$ and let $\mu:\mathcal{N} \to \mathcal{M}$ be an embedding that restricts to an inclusion $\is^{\mathcal{N}} \subseteq \is^{\mathcal{M}}$, $\es^{\mathcal{N}} \subseteq \es^{\mathcal{M}}$. For every $a \in \as^{\mathcal{M}}$ and $i \in \is^{\mathcal{N}}$, we have that $\mu(a)[i]^{\mathcal{M}} = a[i]^{\mathcal{N}}$.
\end{lem}
\begin{proof}
We have that $\mu(a[i]^{\mathcal{N}}) = a[i]^{\mathcal{N}}$ and that $\mu(a[i]^{\mathcal{N}}) = \mu(a)[i]^{\mathcal{M}}$ where the first equality follows from the fact that $\mu$ acts as the inclusion on the element sort and the second from the fact that $\mu$ is an embedding.
\end{proof}

Using the observations above, we are ready to prove that $\tc$ has the amalgamation property.

\begin{thm} \label{thm:tc-amalgam}
The theory $\tc$ has the amalgamation property.
\end{thm}
\begin{proof}
Take two embeddings $\mu_{0}: \mathcal{N} \longrightarrow \mathcal{M}_{0}$ and $\mu_{1}: \mathcal{N} \longrightarrow \mathcal{M}_{1}$. As observed in Corollary~\ref{cor:functional} and Proposition~\ref{prop:inclusion}, we can suppose that $\mathcal{N}, \mathcal{M}_{0}, \mathcal{M}_{1}$ are functional models such that $\mu_{0}, \mu_{1}$ restricts to an inclusion for the sorts $\is$ and $\es$ (so $\mathcal{N}$ is a functional substructure of $\mathcal{M}_0$ and $\mathcal{M}_1$).  

We may assume that $($INDEX$^{\mathcal{M}_{0}} \backslash$INDEX$^{\mathcal{N}}) \cap ($INDEX$^{\mathcal{M}_{1}}\backslash$INDEX$^{\mathcal{N}})=\emptyset$ since we can reduce the general case to this one by embedding $INDEX^{\mathcal{N}}$ into $INDEX^{\mathcal{N}} \cup (INDEX^{\mathcal{M}_0} \cap INDEX^{\mathcal{M}_1})$ and extending the  embeddings $\mu_0$ and $\mu_1$ as that identity over these sets as shown in Figure~\ref{fig:index-embedding}. The homomorphism properties still hold since the signature of sort INDEX is restricted to equality.
\begin{figure}
\begin{tikzcd}[row sep=large, column sep=large]
& \mathrm{INDEX}^{\mathcal{M}} & \\
\mathrm{INDEX}^{\mathcal{M}_0} \ar[ur, dashed] & & \mathrm{INDEX}^{\mathcal{M}_1} \ar[ul, dashed] \\
& \mathrm{INDEX}^{\mathcal{N}}\cup(\mathrm{INDEX}^{\mathcal{M}_0}\cap\mathrm{INDEX}^{\mathcal{M}_1}) \ar[ul, hook'] \ar[ur, hook] & \\
& \mathrm{INDEX}^{\mathcal{N}} \ar[u, hook] \ar[uul, hook] \ar[uur, hook] &
\end{tikzcd}
\caption{Reduction of the amalgamation property over the index sort to the case where $($INDEX$^{\mathcal{M}_{0}} \backslash$INDEX$^{\mathcal{N}}) \cap ($INDEX$^{\mathcal{M}_{1}}\backslash$INDEX$^{\mathcal{N}})=\emptyset$.}
\label{fig:index-embedding}
\end{figure}
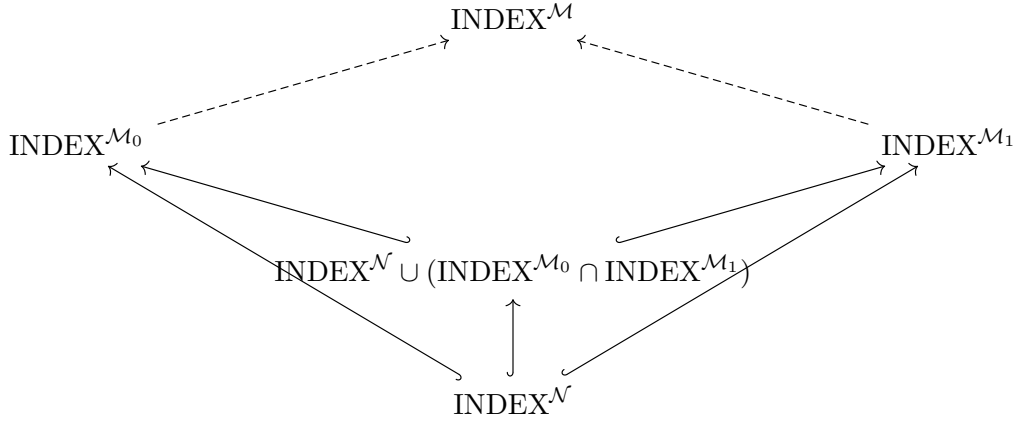

Moreover, we may extend the carrier sets $\mathrm{INDEX}^{\mathcal{M}_0}$ and $\mathrm{INDEX}^{\mathcal{M}_1}$ with new elements $j_i \in INDEX^{\mathcal{M}_i} \setminus INDEX^{\mathcal{N}}$ for $i = 0,1$, again because the INDEX sort uses equality as its only symbol. Thus, the extended models, with the corresponding interpretation of equality, are also models of the theory of equality.

We may assume that $($ELEM$^{\mathcal{M}_{0}} \backslash$ELEM$^{\mathcal{N}}) \cap ($ELEM$^{\mathcal{M}_{1}}\backslash$ELEM$^{\mathcal{N}})=\emptyset$ since we can reduce the general case to this one by embedding $ELEM^\mathcal{N}$ into $ELEM^\mathcal{N} \cup (ELEM^{\mathcal{M}_0} \cap ELEM^{\mathcal{M}_1})$ and extending the  embeddings $\mu_0$ and $\mu_1$ as that identity over these sets. A commutative diagram analogous to that of Figure~\ref{fig:index-embedding} holds for the element sort.  The homomorphism
properties still hold, since the signature of sort ELEM is restricted to uninterpreted functions and relations using a  construction similar to that of the proof of  Proposition~\ref{prop:inclusion}.  

Moreover, we may extend the element carrier sets $\es^{\mathcal{M}_0}$ and $\es^{\mathcal{M}_1}$ with new elements $e_0$ and $e_1$ respectively. The extended structure still satisfies the congruence axioms restricting uninterpreted symbols and relations, as well as the homomorphism conditions, which are unaffected by the newly introduced elements. 

The amalgamated model $\mathcal{M}$ will be defined as the standard model over the index set $\is^{\mathcal{M}_{0}} \cup \is^{\mathcal{M}_{1}}$ and the element set $\es^{\mathcal{M}_0} \cup \es^{\mathcal{M}_1}$. 

Let us define $\nu_{i}: \mathcal{M}_{i} \longrightarrow \mathcal{M}$ for $i=0,1$ in the following manner. For the sorts $\is$ and $\es$ we fix $\nu_i$ to be the inclusion. To define $\nu_{i}$ on $\as^{\mathcal{M}_{i}}$, we extend any $a \in \as^{\mathcal{M}_{i}}$ to indices $k \in \is^{\mathcal{M}_{1-i}} \backslash \is^{\mathcal{N}}$ as follows:
\begin{numcases}{\nu_i(a)(k) :=}
e_{1-i}
& if there is no $c$ such that $\mathcal{M}_i \models |a - \mu_i(c)| < \omega$, \label{I} \\[6pt]
\mu_{1-i}(c)(k)
& for some $c$ satisfying $\mathcal{M}_i \models |a - \mu_i(c)| < \omega$ \label{II}
\end{numcases}

To see that $\nu_i$ is well-defined, we need to check that the choice of $c$ is immaterial in the above definition. In other words, we need to check that for all index $k \in \is^{\mathcal{M}_{1-i}} \setminus \is^{\mathcal{N}}$, $\mu_{1-i}(c)(k) = \mu_{1-i}(c')(k)$.

Since $c$ and $c'$ satisfy that $\mathcal{M}_i \models |a - \mu_i(c)| < \omega$ and $\mathcal{M}_i \models |a - \mu_i(c')| < \omega$ then by transitivity of the cardinality dependence relation, we know that $\mathcal{M}_i \models |\mu_i(c) - \mu_i(c')| < \omega$. Thus, by Lemma~\ref{lem:cardinality}, we have that $\mathcal{N} \models |c-c'| < \omega$. This equivalent to saying that $\mathcal{N} \models c' = \wrte{c}{I}{E}$ where $I \subseteq \is^{\mathcal{N}}$ and $E \subseteq \es^{\mathcal{N}}$. By the properties of the embedding $\mu_{1-i}$, we have that $\mathcal{M}_{1-i} \models \mu_{1-i}(c') = \wrte{\mu_{1-i}(c)}{I}{E}$. Thus, for every $k \in \is^{\mathcal{M}_{1-i}} \setminus \is^{\mathcal{N}}$, $\mu_{1-i}(c')(k) = \mu_{1-i}(c)(k)$, as we wanted to show.

The commutativity condition $\nu_{0} \circ \mu_{0}=\nu_{1} \circ \mu_{1}$ follows trivially for the $\es$ and $\is$ sorts because over these sorts $\mu_0, \mu_1, \nu_0$ and $\nu_1$ act as the inclusion. To show the commutativity property for the sort $\as$ let us proceed applying extensionality. 

If $k \in \is^{\mathcal{N}}$, then 
\begin{align*}
[(\nu_0 \circ \mu_0)(c)](k)
&= [(\nu_0 \circ \mu_0)(c)][k]^{\mathcal{M}} \\
&= \mu_0(c)[k]^{\mathcal{M}} \\
&= c[k]^{\mathcal{N}} \\
&= \mu_1(c)[k]^{\mathcal{M}} \\
&= [(\nu_1 \circ \mu_1)(c)][k]^{\mathcal{M}} \\
&= [(\nu_1 \circ \mu_1)(c)](k)
\end{align*}
where we have used that $\mathcal{M}$ and $\mathcal{N}$ are functional models, the definitions of $\nu_0$ and $\nu_1$ and Lemma~\ref{lem:extension}. 

For indices $k \in \is^{\mathcal{M}_0} \setminus \is^{\mathcal{N}}$,  we have that
\[
[(\nu_0 \circ \mu_0)(c)](k) = [\nu_0(\mu_0(c))] [k]^{\mathcal{M}} = [\mu_0(c)] [k]^{\mathcal{M}_0} = \mu_0(c)(k) = [(\nu_1 \circ \mu_1)(c)](k)
\]
where we have used that $\mathcal{M}$ and $\mathcal{M}_0$ are functional models, Lemma~\ref{lem:extension}. and the definition of $\nu_1$.

The case for $i \in$ INDEX$^{\mathcal{M}_1} \setminus $INDEX$^{\mathcal{N}}$ is analogous.

The injectivity of functions $\nu_0$ and $\nu_1$ follows from the fact that they work by extension of the input argument. Thus, if the image of two inputs is equal, then the inputs themselves must be equal.

To prove that $\nu_0$ and $\nu_1$ are embeddings, we must first complete the definition of the structure $\mathcal{M}$.

In order to define ${\dif{R}^1}^{\mathcal{M}}$ we can simply extend ${\dif{R}^1}^{\mathcal{M}_1} \cup {\dif{R}^1}^{\mathcal{M}_2}$ in such a way that Axiom~\ref{ax:diff} holds. We define ${\dif{R}^1}^{\mathcal{M}}(\overline{a})$ as follows: 
\[
{\dif{R}^1}^{\mathcal{M}}(\overline{a})
=
\begin{cases}
{\dif{R}^1}^{\mathcal{M}_i}(\overline{a}') 
& \text{if for some } i \in \{0,1\},\ \overline{a}=\nu_i(\overline{a}'), \\[6pt]
\text{arbitrary} & \text{if }  \map{R}(\overline{a}) \text{ holds} \\[6pt]
\text{any } i \text{ such that } \lnot R(\overline{a}[i])
& \text{otherwise}  
\end{cases}
\]

The case for iterated diffs is analogous. We define ${\dif{R}^j}^{\mathcal{M}}(\overline{a})$ for $j > 1$ in such a way that Axiom~\ref{ax:iterated} is satisfied:
\[
{\dif{R}^j}^{\mathcal{M}}(\overline{a})
=
\begin{cases}
{\dif{R}^j}^{\mathcal{M}_i}(\overline{a}') 
& \text{if for some } i \in \{0,1\},\ \overline{a}=\nu_i(\overline{a}'), \\[6pt]
{\dif{R}^{j-1}}^{\mathcal{M}}(\overline{a}) & \text{if }  \forall i \notin \{ {\dif{R}^1}^{\mathcal{M}}(\overline{a}), \ldots, {\dif{R}^{j-1}}^{\mathcal{M}}(\overline{a})   \}. R(\overline{a}[i]) \\[6pt]
\text{any } i \text{ such that } \lnot R(\overline{a}[i])
& \text{otherwise}  
\end{cases}
\]

To prove that ${\dif{R}^j}^{\mathcal{M}}$ is well-defined, we show that 
\begin{align} \label{eq:claim}
\text{If } \overline{a}=\nu_{0}\!\left(\overline{a_{0}}\right)
=\nu_{1}\!\left(\overline{a_{1}}\right), 
\text{then there exists } \overline{c}
\text{ such that }
\overline{a_{0}}=\mu_{0}(\overline{c})
\text{ and }
\overline{a_{1}}=\mu_{1}(\overline{c}).
\end{align}

Given this property, we would have that 
\[
{\dif{R}^j}^{\mathcal{M}_i}(\overline{a_i}) = {\dif{R}^j}^{\mathcal{M}_i}(\mu_i(\overline{c})) = \mu_i({\dif{R}^j}^{\mathcal{N}}(\overline{c})) = {\dif{R}^j}^{\mathcal{N}}(\overline{c})
\]
where we have used the fact that $\mu_i$ are homomorphisms that it acts as the inclusion over the index sort. As a consequence, we have that 
\[
{\dif{R}^j}^{\mathcal{M}_0}(\overline{a_0}) = {\dif{R}^j}^{\mathcal{M}_1}(\overline{a_1})
\]
so the first case of the definition of the operator $\operatorname{diff}_R^{\mathcal{M}}$ is not ambiguous. 

\subsubsection*{Proof of Claim~(\ref{eq:claim}).}

Suppose that $\overline{a}=\nu_{0}\left(\overline{a_{0}}\right)=\nu_{1}\left(\overline{a_{1}}\right)$. $\nu_{0}\left(\overline{a_{0}}\right)$ and $\nu_{1}\left(\overline{a_{1}}\right)$
must have been defined as in case~(\ref{II}) above (otherwise they cannot coincide with each other at indexes $\left.j_{0}, j_{1}\right)$. This means that for each $i \in \{0,1\}$ there exists a tuple of arrays $\overline{c_{i}}$ such that  we have $\mathcal{M}_{i} \models\left|\overline{a_{i}}-\mu_{i}\left(\overline{c_{i}}\right)\right|< \omega$. 

Since $\nu_{0}\left(\overline{a_{0}}\right)=\overline{a}=\nu_{1}\left(\overline{a_{1}}\right)$, this means that the arrays $\nu_{0}\left(\mu_{0}\left(\overline{c_{0}}\right)\right)=\nu_{1}\left(\mu_{1}\left(\overline{c_{0}}\right)\right)$ (where the equality follows from the commutativity condition $\nu_{0} \circ \mu_{0}=\nu_{1} \circ \mu_{1}$) and $\overline{a}$ differ only at finitely many indices (because $\overline{a_0}$ and $\mu_0(\overline{c_0})$ differ in a finite number of positions and are completed by $\nu_0$ in the same way). The same is true for $\nu_{1}\left(\mu_{1}\left(\overline{c_{1}}\right)\right)$ and $\overline{a}$, so that $\nu_{1}\left(\mu_{1}\left(\overline{c_{0}}\right)\right)$ and $\nu_{1}\left(\mu_{1}\left(\overline{c_{1}}\right)\right)$ differ only at finitely many indices. 

The same consequently holds for $\overline{c_{0}}, \overline{c_{1}}$ in $\mathcal{N}$ too, for $\mu_{0}\left(\overline{c_{0}}\right)$ and $\mu_{0}\left(\overline{c_{1}}\right)$ in $\mathcal{M}_{0}$ and for $\mu_{1}\left(\overline{c_{0}}\right)$ and $\mu_{1}\left(\overline{c_{1}}\right)$ in $\mathcal{M}_{1}$. This is because embeddings keep adding indices and elements so if in the co-domain they differ in a finite number of positions then they do it too on the input domain. 

Since the choice of $c$ in case~(\ref{II}) above is immaterial, we can suppose that $\overline{c_{0}}=\overline{c_{1}}$. Let us use $\overline{c}$ to name it. By case~(\ref{II}), $\overline{a} = \nu_{1}\left(\overline{a_{1}}\right)$ and $\nu_{0}\left(\mu_{0}(\overline{c})\right)) = \nu_{1}\left(\mu_{1}(\overline{c})\right)$ cannot differ at any $k \in $ ELEM$^{\mathcal{M}_{0}} \backslash$ ELEM$^{\mathcal{N}}$ since $\nu_1(\overline{a_1})(k) = \mu_0(\overline{c})(k)$. Similarly, $\nu_{0}\left(\mu_{0}(\overline{c})\right)=\nu_{1}\left(\mu_{1}(\overline{c})\right)$ and $\overline{a} = \nu_0(\overline{a_0})$ cannot differ at any $k \in \mathrm{ELEM}^{\mathcal{M}_{1}} \backslash \mathrm{ELEM}^{\mathcal{N}}$ since $\nu_0(\overline{a_0})(k) = \mu_1(\overline{c})(k)$. Thus $\overline{a}$ and $\nu_{0}\left(\mu_{0}(\overline{c})\right)=\nu_{1}\left(\mu_{1}(\overline{c})\right)$ possibly differ only for finitely many $k \in \operatorname{INDEX}^{\mathcal{N}}$. 

Since $\overline{a}=\nu_{0}\left(\overline{a_{0}}\right)=\nu_{1}\left(\overline{a_{1}}\right)$, the values of $\overline{a}$ at any $k \in$ INDEX$^{\mathcal{N}}$ belong to ELEM$^{\mathcal{M}_{0}} \cap$ ELEM$^{\mathcal{M}_{1}}=$ ELEM$^{\mathcal{N}}$ (the embeddings only extend the arrays). This means that 
\[
\overline{a} = \wrte{\nu_{0}\left(\mu_{0}(\overline{c})\right)}{I}{E}^{\mathcal{M}}= \nu_{0}\left(\mu_{0}\left(\wrte{\overline{c}}{I}{E})^{\mathcal{N}}\right)\right)
\]
for $I \subseteq \operatorname{INDEX}^{\mathcal{N}}$ and $E \subseteq \operatorname{ELEM}^{\mathcal{N}}$. Thus, we have that $\overline{a}$ is of the kind $\nu_{0}\left(\mu_{0}(\overline{d})\right)=\nu_{1}\left(\mu_{1}(\overline{d})\right)$ and from $\overline{a}=\nu_{0}\left(\overline{a_{0}}\right)=\nu_{1}\left(\overline{a_{1}}\right)$, we get $\overline{a_{0}}=\mu_{0}(\overline{d})$ and $\overline{a_{1}}=\mu_{1}(\overline{d})$ (because $\nu_{0}, \nu_{1}$ are injective), as was to be shown. \qed

\vspace{1em}

To conclude that ${\dif{R}^{j}}^{\mathcal{M}_i}(\overline{a}_i)$ is well-defined, it remains to prove that its second case is well-defined. But this follows by strong induction argument on $j$ (i.e. using the fact that all of the terms ${\dif{R}^1}^{\mathcal{M}}(\overline{a}), \ldots, {\dif{R}^{j-1}}^{\mathcal{M}}(\overline{a})$ are well-defined as inductive hypothesis).

Let us define the corresponding relations $R$ and 
$\map{R}$ on the structure $\mathcal{M}$. With start with relations $R$ over element tuples:
\[
{R}^{\mathcal{M}}(\overline{e})
=
\begin{cases}
{R}^{\mathcal{M}_i}(\overline{e}) 
& \text{if for some } i \in \{0,1\},\ \overline{e} \in (\es^{\mathcal{M}_i})^{\ar{R}}, \\[6pt]
\text{true}
& \text{otherwise}  
\end{cases}
\]
where $\ar{R}$ is the arity (or number of arguments) of the relation symbol $R$.

There is no ambiguity in the first case of this definition since $\overline{e} \in (\es^{\mathcal{M}_0})^{ar(R)} \cap (\es^{\mathcal{M}_1})^{ar(R)}$ if and only if $\overline{e} \in (\es^{\mathcal{N}})^{ar(R)}$. Given that $\mu_0$ and $\mu_1$ are embeddings acting as the inclusion over the element sort we have that 
\[
R^{\mathcal{M}_0}(\overline{e}) \iff R^{\mathcal{M}_0}(\mu_0(\overline{e})) \iff R^{\mathcal{N}}(\overline{e}) \iff R^{\mathcal{M}_1}(\mu_1(\overline{e})) \iff R^{\mathcal{M}_1}(\overline{e}) 
\]

Similarly, we define:
\[
\map{R}^{\mathcal{M}}(\overline{a})
=
\begin{cases}
\map{R}^{\mathcal{M}_i}(\overline{a}') 
& \text{if for some } i \in \{0,1\},\ \overline{a}=\nu_i(\overline{a}'), \\[6pt]
\text{true}
& \text{otherwise}  
\end{cases}
\]
There is no ambiguity in the first case of this definition since if $\overline{a} = \nu_0(\overline{a}_0) = \nu_1(\overline{a}_1)$ then, by Claim~(\ref{eq:claim}), we have that there exists a tuple of arrays $\overline{c}$ such that $\overline{a}_0 = \mu_0(\overline{c})$ and $\overline{a}_1 = \mu_1(\overline{c})$. Thus, 
\[
\map{R}^{\mathcal{M}_0}(\overline{a}_0) \iff \map{R}^{\mathcal{M}_0}(\mu_0(\overline{c})) \iff \map{R}^{\mathcal{N}}(\overline{c}) \iff \map{R}^{\mathcal{M}_1}(\mu_1(\overline{c})) \iff \map{R}^{\mathcal{M}_1}(\overline{a}_1)
\]

This finishes the definition of the structure $\mathcal{M}$. To prove that $\nu_0$ and $\nu_1$ are embeddings it remains to be shown that they preserve function symbols and preserve and reflect intepreted predicates.

For diff terms we have that
\[
{\dif{R}^j}^{\mathcal{M}}(\nu_i(\overline{a})) = {\dif{R}^j}^{\mathcal{M}_i}(\overline{a}) = \nu_i({\dif{R}^j}^{\mathcal{M}_i}(\overline{a}))
\]
where the first equality follows from the definition of the diff operators and the second equality follows from the fact that $\nu_i$ is defined as the inclusion over the index sort.

For relation symbols, we have that
\[
R^{\mathcal{M}}(\nu_i(\overline{e})) \iff R^{\mathcal{M}}(\overline{e}) \iff R^{\mathcal{M}_i}(\overline{e})
\]
where the first equivalence follows from the fact that $\nu_i$ acts as the inclusion over the element sort and the second equivalence follows from the definition of relations in the structure $\mathcal{M}$ since $\overline{e} \in (\es^{\mathcal{M}_i})^{ar(R)}$.

For map symbols, we have that 
\[
\map{R}^{\mathcal{M}}(\nu_i(\overline{a})) \iff \map{R}^{\mathcal{M}_i}(\overline{a}) 
\]
by the definition of $\map{R}^{\mathcal{M}}$.

With this, the proof of each $\nu_i$ being an embedding is concluded and thus, the amalgamation property for $\tc$ is established.
\end{proof}

\begin{cor}
The theory $\tc$ admits quantifier-free interpolation.
\end{cor}
\begin{proof}
It follows from Theorems~\ref{thm:amalg-interp} and \ref{thm:tc-amalgam}.
\end{proof}

\section{General Interpolation Properties}
\label{section:gqfinterp}
We start by recalling the definition of general quantifier-free interpolation from \cite{bruttomesso_quantifier-free_2014}.

\begin{defi}[General Quantifier-Free Interpolation]
Let $T$ be a theory in a signature we say that $T$ has the general quantifier-free interpolation property if and only if for every signature $\Sigma'$ (disjoint from $\Sigma$) and for every pair of ground $\Sigma \cup \Sigma'$-formulas $\phi, \psi$ such that $\phi \land \psi$ is 
$T$-unsatisfiable, there is a ground formula $\theta$ such that (i) $\phi$ $T$-entails $\theta$; (ii) $\theta \land \psi$ is $T$-unsatisfiable; (iii) all predicate, constants and function  symbols from $\Sigma'$ occurring in $\theta$ also occur in $\phi$ and $\psi$.
\end{defi}

The notion of strong amalgamation \cite{bruttomesso_strong_2012} extends the amalgamation property of Definition~\ref{def:amalgamation} as follows.

\begin{defi}
A theory $T$ has the strong amalgamation property if the preceding embeddings $\mu_1$, $\mu_2$ and the preceding model $\mathcal{M}$ can be chosen to satisfy the following additional condition: if for some $m_1$, $m_2$ we have $\mu_1(m_1) = \mu_2(m_2)$, then there exists an element $a$ in $\mathcal{N}$ such that $m_1 = a = m_2$. 
\end{defi}

Note that here, unlike in \cite{bruttomesso_strong_2012}, we speak of strong amalgamation and not sub-amalgamation, since the theories of interest are universal. In particular, this implies that $\mathcal{N}$ is a model of $T$. 

\begin{cor} \label{cor:tc-strong}
The theory $\tc$ has the strong amalgamation property.
\end{cor}
\begin{proof}
The result follows from the proof of Theorem~\ref{thm:tc-amalgam}.
\end{proof}

It turns out that the notion of strong amalgamation is sufficient for general quantifier-free interpolation. 

\begin{thm}[\cite{bruttomesso_quantifier-free_2014}]
A theory has the general quantifier free interpolation property if and only if it is strongly sub-amalgamable.
\end{thm}

Since for universal theory sub-amalgamability implies strong amalgamation \cite{ghilardi_modularity_2018}, we have the following Corollary.



\begin{cor}
The theory $\tc$ has the general quantifier-free interpolation property.
\end{cor}
\begin{proof}
The result follows from Corollary~\ref{cor:tc-strong}.
\end{proof}

\section{Uniform Interpolation Properties}
\label{section:unifinterp}
We will now study the notion of uniform interpolation for our logics of interest. We will see that, at least in the set of examples that we study, uniform interpolation holds if and only if the logic admits quantifier-elimination procedures. To compute the corresponding interpolants, we will thus only need to run elimination of existential quantifiers.

Let us start with the definition of uniform interpolant or cover, which was introduced in \cite{gulwani_cover_2008}.

\begin{defi}[Uniform interpolation]
\label{def:unifinterp}
Let $T$ be a theory and, given an existential formula $\exists \overline{e}. \phi(\overline{e}, \overline{y})$, the set of residues of $\exists \overline{e}. \phi(\overline{e}, \overline{y})$ is
$\operatorname{Res}(\exists \overline{e}. \phi(\overline{e}, \overline{y}))=\{\theta(\overline{y}, \overline{z}) \mid T \models \exists \overline{e}. \phi(\overline{e}, \overline{y}) \rightarrow \theta(\overline{y}, \overline{z}) \text{ and } \theta(\overline{y}, \overline{z}) \text{ is quantifier-free} \}$.
A formula $\psi(\overline{y})$ is a $T$-cover of $\exists \overline{e}. \phi(\overline{e}, \overline{y})$ if and only if $\psi(\overline{y}) \in \operatorname{Res}(\exists \overline{e}. \phi(\overline{e}, \overline{y}))$ and $\psi(\overline{y})$ implies (modulo $T)$ all the other formulas in $\operatorname{Res}(\exists \overline{e}. \phi(\overline{e}, \overline{y}))$. The theory $T$ has the \textit{uniform quantifier-free interpolation} property if every existential formula $\exists \overline{e}. \phi(\overline{e}, \overline{y})$ has a $T$-cover.
\end{defi}

Existential quantifier elimination yields uniform interpolants.

\begin{prop}
\label{prop:unifqe}
If a theory $T$ eliminates existential quantifiers then it has the uniform quantifier-free interpolation property. 
\end{prop}
\begin{proof}
Let $T$ be a theory that eliminates existential quantifiers. By Definition~\ref{def:unifinterp}, the theory has uniform interpolation if and only if for every formula $\exists \overline{e}. \phi(\overline{e}, \overline{y})$, there exists a formula $ \theta(\overline{y},\overline{z})$, such that $T \models \exists \overline{e}. \phi(\overline{e},\overline{y}) \to \theta(\overline{y},\overline{z})$ and for any formula $\theta'(\overline{y},\overline{z})$, such that $T \models \exists \overline{e}. \phi(\overline{e},\overline{y}) \to \theta'(\overline{y},\overline{z})$, we get that $T \models \theta(\overline{y},\overline{z}) \to \theta'(\overline{y},\overline{z})$. Now, existential quantifier elimination implies that for $\exists \overline{e}. \phi(\overline{e}, \overline{y})$, we have an equivalent formula in the theory without quantifiers $\phi'(\overline{y})$. We can then take $\theta(\overline{y},\overline{z}) := \phi'(\overline{y})$ as a $T$-cover.
\end{proof}

\subsection{Failure of Uniform Interpolation for Combinatory Array Logic}

The following counterexample has been provided by Ghilardi in unpublished work \cite{ghilardi_interpolation_2024}. We present it for the sake of completeness. The example is formulated for the extensional theory of arrays with diff, which is subsumed by combinatory array logic with diffs.

Let us start defining the notions of ultrafilter and ultraproduct, since they are needed for the formulation of the counterexample.

\begin{defi}[\cite{keisler_ultraproduct_2010}] \label{def:ultraproduct}
If $I$ is a non-empty set, a filter over $I$ is a set $D$ of subsets of $I$ such that:

\begin{itemize}
    \item $\emptyset \notin D, I \in D$.
    \item If $X,Y \in D$ then $X \cap Y \in D$.
    \item If $X \in D$ and $X \subseteq Y \subseteq I$ then $Y \in D$.
\end{itemize}

An ultrafilter is a filter that is maximal with respect to inclusion.

Let $\mathcal{F}$ be a filter.

The reduced product of a collection of models $\mathcal{A}_i = \langle A_i,\ldots \rangle$ with $i \in I$ modulo $\mathcal{F}$ is  a structure $\mathcal{A} = \prod_{i \in I} \mathcal{A}_i / \mathcal{F}$ whose domain is the set $\{ g: I \to \cup_{i \in I} A_i \, \mid \, \forall i \in I. g(i) \in A_i \} / N$ where $N$ is the equivalence relation given by 
\[
g N g' \iff \{ i \in I \, \mid \, g(i) = g'(i) \} \in \mathcal{F}
\]
We write $[g]_{N}$ for the equivalence class of function $g$ modulo the equivalence relation $N$. When $N$ is clear from the context, we write $[g]$ omitting $N$.

Each constant symbol $c$ is interpreted as $c^{\mathcal{A}} = [ \{ g: I \to \cup_{i \in I} A_i \, \mid \, \forall i \in I. g(i) = c^{A_i} \} ]_{N}$.

For each function symbol $f$, $f^{\mathcal{A}}([g_1]_{N}, \ldots, [g_k]_{N}) = [g_0]_{N}$ where for each $i \in I$, 
\[
g_0(i) = f^{\mathcal{A}_i}(g_1(i),\ldots,g_k(i))
\]

For each relation symbol $R$, $R^{\mathcal{A}}([g_1]_N, \ldots,[g_k]_N)$ holds if and only if
\[
\{ i \in I \mid R^{\mathcal{A}_i}(g_1(i),\ldots,g_k(i)) \} \in \mathcal{F}
\]
If $\mathcal{F}$ is an ultrafilter, then the resulting structure is called an ultraproduct.
\end{defi}

We will not show that the notions given in Definition~\ref{def:ultraproduct} are independent of the chosen class representative as this proof is standard. One key result that we use is Łoś's theorem. 

\begin{thm}[Łoś]
Let $\varphi(x_1,\ldots,x_k)$ be a first-order formula and $[g_1],\ldots,[g_k]$ be elements of $\mathcal {A} = \prod_{i \in I} \mathcal{A}_i / \mathcal{U}$ where $\mathcal{U}$ is an ultrafilter.

Then, $\mathcal{A} \models \varphi([g_1],\ldots,[g_k])$ if and only if $\{ i \in I | \mathcal{A}_i \models \varphi(g_1(i),\ldots,g_k(i)) \} \in \mathcal{U}$.
\end{thm}

We can now proceed with the counterexample itself.

\begin{prop} \label{prop:ghilardi}
The formula 
\[
G(c_1,c_2,d_1,d_2,i) := c_1[i] \neq c_2[i] \land d_1[i] = d_2[i]
\]
cannot eliminate $i$ uniformly in  combinatory array logic with diff functions.
\end{prop}
\begin{proof}
By contradiction, assume the existence of a uniform interpolant of $G$, denoted by $UI(c_1,c_2,d_1,d_2)$. Consider the formula
\[
\phi_n(c_1,c_2,d_1,d_2) := c_1 \sim_n c_2 \to \bigvee_{j = 1}^n d_1[\diff{=}{j}{c_1,c_2}] = d_2[\diff{=}{j}{c_1,c_2}]
\]
where $c_1 \sim_n c_2$ encodes the property that $c_1$ and $c_2$ differ in at most $n$ indices, which can be written in our context as
\[
\diff{=}{n+1}{c_1,c_2} = \diff{=}{n}{c_1,c_2}
\]
Thus, $\phi_n$ encodes the property that if $c_1$ and $c_2$ differ in at most $n$ positions then $d_1$ and $d_2$ must coincide in some of them. Clearly, 
\begin{align} \label{eq:uniform}
\exists i. G(c_1,c_2,d_1,d_1,i) \to \phi_n(c_1,c_2,d_1,d_2)
\end{align}
is a valid formula for all $n$. By definition of uniform interpolant, it follows that 
\begin{align} \label{eq:impl}
UI(c_1,c_2,d_1,d_2) \to \phi_n(c_1,c_2,d_1,d_2)   
\end{align}
is also valid.

We reach a contradiction by building a model satisfying both $UI(c_1,c_2,d_1,d_2)$ and $\lnot UI(c_1,c_2,d_1,d_2)$. Let $\mathcal{M}_n$ (for $n > 0$) be models such that 
\[
\mathcal{M}_n \not\models \phi_n(c_1,c_2,d_1,d_2)
\]
For instance, take $c_1,c_2$ differing in exactly $n$ positions and let $d_1:=c_1, d_2:=c_2$. Thus, $\mathcal{M}_n \models \lnot \phi_n(c_1,c_2,d_1,d_2)$ and by the validity of Formula~(\ref{eq:impl}), $\mathcal{M}_n \models \lnot UI(c_1,c_2,d_1,d_2)$. 

One can then take an ultraproduct over on the natural numbers $\mathcal{U}$ to obtain that 
\begin{align} \label{eq:ultra}
\Pi_{n \in \mathbb{N}} \mathcal{M}_n / \mathcal{U}  \models \lnot UI(c_1,c_2,d_1,d_2)
\end{align}
where for each $n$ and each vector $v \in \{c_1,c_2,d_1,d_2\}$, we define $v(n)$ as the corresponding vector satisfying $\lnot UI(c_1,c_2,d_1,d_n)$ in $\mathcal{M}_n$. By Łoś's theorem, Judgement~(\ref{eq:ultra}) follows from the fact that
\[
\{ n \mid \mathcal{M}_n \models \lnot UI(c_1(n),c_2(n),d_1(n),d_2(n)) \} \in \mathcal{U}
\]
since a non-empty ultrafilter always contains the universal set. For similar reasons, we also have that $\Pi_{n \in \mathbb{N}} \mathcal{M}_n / \mathcal{U}$ is a model of $\tc$. 

Now consider an extension $\mathcal{N} \supseteq \Pi_{n \in \mathbb{N}} \mathcal{M}_n / \mathcal{U}$ defined by adding a new index $j$ to the index set of the ultraproduct. Since the index set of $\Pi_{n \in \mathbb{N}} \mathcal{M}_n / \mathcal{U}$ is infinite, we are free to define $c_1(j), c_2(j), d_1(j)$ and $d_2(j)$ such that $c_1(j) \neq c_2(j)$ and $d_1(j) = d_2(j)$ without affecting the previous argument about the formulas $\phi_n$. Thus, $\mathcal{N} \models \exists i. G(c_1,c_2,d_1,d_2,i)$ and by Formula~(\ref{eq:uniform}), we have that $\mathcal{N} \models UI(c_1,c_2,d_1,d_2)$. Finally, since quantifier-free formulas are preserved by substructures (see Corollary~\ref{cor:preserv}), it follows that $\Pi_{n \in \mathbb{N}} \mathcal{M}_n / \mathcal{U} \models UI(c_1,c_2,d_1,d_2)$, a contradiction.
\end{proof}

\begin{cor}
Combinatory array logic does not have the uniform interpolation property.
\end{cor}

\subsection{Uniform Interpolation for the Theory of Sets with Cardinalities}
\label{section:uibapa}

Uniform interpolation follows from a combination of the quantifier-elimination procedures of Zarba \cite{zarba_quantifier_2004}, Revesz \cite{revesz_quantifier-elimination_2004} and later Rinard-Kun\v{c}ak \cite{kuncak_algorithm_2005}. Since the theory eliminates quantifiers, it follows that it also eliminates existential quantifiers. The result of eliminating existential quantifiers is thus the desired uniform interpolant (see Proposition~\ref{prop:unifqe}). We now sketch the existential quantifier-elimination procedure for the case of our signature.

\subsubsection*{Existential quantifier-elimination procedure}

As it is well-known it is sufficient to show how to eliminate each existential quantifier from a conjunction of literals of the theory, since, otherwise, we can convert the formula into disjunctive normal form and distribute the existential quantifier in each disjunct (see \cite[Lemma~2.7.4]{hodges_model_1993}). Now, by the structure of the syntax tree of the formulas from Figure~\ref{fig:alberti-syntax} (without set comprehensions), it follows that our formula will be of the form
\[
\varphi_1(i_1,\ldots,i_n, B_1,\ldots,B_s) \land \varphi_2(k_1,\ldots,k_m,|B_1|,\ldots,|B_s|) \land \varphi_3(S_1,\ldots,S_n)
\]
where $i$ variables are indices, $k$ variables are integers, $B$ are Boolean algebra expressions and $S$ variables are set variables. Here, $\varphi_1$ takes care of the first two terms of the BNF grammar under symbol $A$ (Figure~\ref{fig:alberti-syntax}), $\varphi_2$ refers to the third and fourth terms under that symbol and $\varphi_3$ refers to the last three expressions under that symbol. 

Now, suppose that we want to eliminate an integer variable $k$.
\[
\exists k. \varphi_1(i_1,\ldots,i_n, B_1,\ldots,B_s) \land \varphi_2(k_1,\ldots,k_m,|B_1|,\ldots,|B_s|) \land \varphi_3(S_1,\ldots,S_n)
\]
We can eliminate this variable by rewriting the formula into
\begin{align*}
\exists k. \exists l_1,\ldots,l_s \ge 0. \varphi_1(i_1,\ldots,i_n,B_1,\ldots,B_s) \land & \varphi_2(k_1,\ldots,k_m,l_1,\ldots,l_s) \land \\ & \varphi_3(S_1,\ldots,S_n) \land \bigwedge_{i = 1}^s |B_i| = l_i
\end{align*}
and noting that the variable $k$ only appears in $\varphi_2$. From which, we can rewrite it to
\begin{align*}
\exists l_1,\ldots,l_s \ge 0. \varphi_1(i_1,\ldots,i_n, B_1,\ldots,B_s) \land & \exists k. \varphi_2(k_1,\ldots,k_m,l_1,\ldots,l_s) \land \\ & \varphi_3(S_1,\ldots,S_n) \land \bigwedge_{i = 1}^s |B_i| = l_i
\end{align*}
Finally, we can apply Presburger's quantifier elimination procedure in the formula 
\[
\exists k. \varphi_2(k_1,\ldots,k_m,l_1,\ldots,l_s)
\]
which yields an equivalent formula not containing $k$ and apply the one-point rule \cite{gries_logical_1993} to eliminate the existential quantifiers in $l_1,\ldots,l_s$. Note that after this step, since we can exchange the order of existential quantifiers (and we only need to eliminate these), we can assume that there are no more existential quantifiers over integer variables. 

Now, suppose that we want to eliminate an index variable $i$. 
\[
\exists i. \varphi_1(i_1,\ldots,i_n, B_1,\ldots,B_s) \land \varphi_2(k_1,\ldots,k_m,|B_1|,\ldots,|B_s|) \land \varphi_3(S_1,\ldots,S_n)
\]
Since $i$ may only occur in $\varphi_1$ it follows that it is enough to eliminate $i$ in the formula 
\[
\exists i. \varphi_1(i_1,\ldots,i_n, B_1,\ldots,B_s)
\]
Now, this case is effectively reducible to the case of elimination of quantifiers in a formula of type $\varphi_2$ by applying the following syntactic transformations: we rewrite index variables into singleton sets, disequalities between indices into disequalities between singleton sets and membership as subset relations between sets. We then apply the elimination procedure that we  give for $\varphi_2$ and $\varphi_3$ and conclude by rewriting the introduces singleton sets into indices, obtaining an equivalent formula. 

Finally, assume that we want to eliminate a set variable $S$ from the formula.
\[
\exists S. \varphi_2(k_1,\ldots,k_m,|B_1|,\ldots,|B_s|) \land \varphi_3(S_1,\ldots,S_n)
\]
Here we observe that the relations in $\varphi_3$ can be encoded in $\varphi_2$. If we have a relation $S_1 \subseteq S_2$ this is equivalent to the condition $|S_1 \setminus S_2| = 0$. Similarly, the condition $S_1 \not\subseteq S_2$ can be encoded by $|S_1 \setminus S_2| > 0$. Finally, the condition $S_1 = S_2$ is equivalent to $S_1 \subseteq S_2$ and $S_1 \supseteq S_2$. Thus, in what follows we may only be concerned with eliminating the existential quantifier from the formula
\[
\exists S. \varphi_2(k_1,\ldots,k_m,|B_1|,\ldots,|B_s|) 
\]

Again, we introduce existentially quantified variables for the cardinalities of the Boolean algebra expressions and move the existential quantification over $S$ to the part of the formula using Boolean algebra expressions, obtaining 
\[
\exists l_1,\ldots,l_s \ge 0.  \varphi_2(k_1,\ldots,k_m,l_1,\ldots,l_s) \land \exists S. \bigwedge_{i = 1}^s |B_i| = l_i 
\]

Now, as in \cite{kuncak_algorithm_2005}  we rewrite the equation $|B_i| = l_i$ as 
\[
\exists l_{\beta_1},\ldots,l_{\beta_{2^n}} \ge 0. \sum_{j = 1}^{2^n} \llbracket B_i \rrbracket_{\beta_j} \cdot l_{\beta_j} = l_i \land \bigwedge_{j = 1}^{2^n} |E_{\beta_j}| = l_{\beta_j}
\]
thus, the original formula can be rewritten into
\begin{align*}
\exists l_1,\ldots,l_s, l_{\beta_1},\ldots,l_{\beta_{2^n}} \ge 0.  & \varphi_2(k_1,\ldots,k_m,l_1,\ldots,l_s, l_{\beta_1},\ldots,l_{\beta_{2^n}}) \land \\ & \exists S. \bigwedge_{j = 1}^{2^n} |E_{\beta_j}| = l_{\beta_j}
\end{align*}

We now need to eliminate the occurrences of $S$ in the elementary Venn regions $E_{\beta_j}$. The solution is to pair together elementary Venn regions such that $\beta(S) = 1$ and $\beta(S) = 0$ but are otherwise identical (for instance if $S := S_1$ then $\beta=1x$ and $\beta' = 0x$ where $x \in \{0,1\}^*$, would be paired together). Now, one reorders the Venn regions in a way that the element $2j$ and $2j-1$ are couples of such pairs. The list of elementary Venn regions where the variable $S$ is eliminated is indexed by $E_{\beta_j}'$. Then, the above formula is equivalent to
\begin{align*}
\exists l_1,\ldots,l_s, l_{\beta_1},\ldots,l_{\beta_{2^n}} \ge 0.  & \varphi_2(k_1,\ldots,k_m,l_1,\ldots,l_s, l_{\beta_1},\ldots,l_{\beta_{2^n}}) \land \\ & \bigwedge_{j = 1}^{2^{n-1}} |E_{\beta_j}'| = l_{\beta_{2j}}+l_{\beta_{2j-1}}
\end{align*}

One then iterates the process by eliminating the existentially quantified integer variables. Note that when there is only one set variable, the constraint reduces to the satisfiability of the Presburger arithmetic subformula.

\begin{cor}
The theory of sets with cardinalities is uniform interpolating. 
\end{cor}
\begin{proof}
The proof follows from the quantifier-elimination procedure above and Proposition~\ref{prop:unifqe}.
\end{proof}

\begin{exa}
Let us show how to eliminate existential quantifiers in the expression $i \in S$. First, we consider how to eliminate the existential quantifier from $\exists i. i \in S$, then how to eliminate the existential quantifier from $\exists S. i \in S$. 

In the first case, the formula is translated into $\exists \{i\}. |\{i\} \setminus S| = 0 \land |\{i\}| = 1$. Solving in the cardinality equations, we get that $l_{(0,1)} = 0$ and $l_{(1,1)} = 1$ with $l_{(1,0)}, l_{(1,1)}$ non-negative. Thus, the system of equations $|S^c| = l_{(0,0)}$, $|S| = 1 + l_{(1,0)}$ is equivalent to the formula $|S| \ge 1$. 

In the second case, the formula is translated to $\exists S. |\{i\} \setminus S| = 0 \land |\{i\}| = 1$. Reusing the cardinality equations of the previous paragraph, and solving for $\{i\}$, we have that the formula is equivalent to $|\{i\}| = 1 \land |\{i\}^c| \ge 0$, which is trivially true. Thus, the formula is equivalent to the formula $i = i$. 
\end{exa}

\subsection{Uniform Interpolation for the Simple Flat Array Fragment}
\label{section:sff}

That the simple flat array fragment eliminates existential quantifiers follows from an application of the quantifier elimination techniques from \cite{feferman_first_1959}, which have been later rediscovered in \cite{ghilardi_higher-order_2020}. Again, these techniques must be adapted to the specific language of Figure~\ref{fig:alberti-syntax}. The technique extends the method of the previous section with a methodology to eliminate existentially quantified array variables. This further step requires us to witness that certain Venn regions are non-empty, which is done with existentially quantified formulas of the element theory. We now sketch the quantifier-elimination procedure for the case of our signature.

\subsubsection*{Existential quantifier-elimination procedure}

As in Section~\ref{section:uibapa}, it is sufficient to show how to eliminate each existential quantifier from a conjunction of literals of the theory. By the structure of the syntax tree of the formulas from Figure~\ref{fig:alberti-syntax}, it follows that our formula will be of the form
\begin{align}
\label{eq:original}
\begin{split}
& \varphi_1(i_1,\ldots,i_n, B_1,\ldots,B_s) \land \varphi_2(k_1,\ldots,k_m,|B_1|,\ldots,|B_s|) \land \\ & \varphi_3(S_1,\ldots,S_n) \land \bigwedge_{j = l}^{n} S_j = \{ i \mid \varphi_j(a_1[i],\ldots,a_t[i], \overline{c}) \}
\end{split}
\end{align}
where $i$ variables are indices, $k$ variables are integers, $B$ are Boolean algebra expressions, $S$ variables are set variables and $c$ are element theory constants. Here, $\varphi_1$ takes care of the first two terms of the BNF grammar under symbol $A$ (Figure~\ref{fig:alberti-syntax}), $\varphi_2$ refers to the last three expressions under that symbol, and $\varphi_3,\varphi_4$ refer to the third and fourth expressions under that symbol. 

Now, suppose that we want to eliminate an array variable $a$.
\begin{align*}
\exists a. & \varphi_1(i_1,\ldots,i_n, B_1,\ldots,B_s) \land \varphi_2(k_1,\ldots,k_m,|B_1|,\ldots,|B_s|) \land \\ & \varphi_3(S_1,\ldots,S_n) \land \bigwedge_{j = l}^{n} S_j = \{ i \mid \varphi_j(a_1[i],\ldots,a_t[i], \overline{c}) \}
\end{align*}
without loss of generality assume that $a := a_1$. 

Each model of this formula determines a unique model of the following and viceversa.
\begin{align}
\label{eq:qelimfv}
\begin{split}
\exists a. & \varphi_1(i_1,\ldots,i_n, B_1,\ldots,B_s) \land \varphi_2(k_1,\ldots,k_m,|B_1|,\ldots,|B_s|) \land \\ & \varphi_3(S_1,\ldots,S_n) \land \bigwedge_{i = 1}^n S_i = \bigcup_{\beta \models S_i} S^{\beta} \land \\ & \bigwedge_{\beta = 1}^{2^{n}} S^{\beta} = \{ i \mid i \in S_1^{\beta(1)} \land \ldots \land i \in S_{l-1}^{\beta(l-1)} \land \varphi^{\beta(l..n)}(a[i],\ldots,a_t[i], \overline{c}) \}
\end{split}
\end{align}
where the advantage is that the formulas in the set comprehensions form a partition of the index set.

We will show that this formula is equivalent to the following formula
\begin{align}
\label{eq:afterqefv}
\begin{split}
\exists &S'^{\beta_1},\ldots,S'^{\beta_{2^n}}.  \varphi_1(i_1,\ldots,i_n, B_1',\ldots,B_s') \land \varphi_2(k_1,\ldots,k_m,|B_1'|,\ldots,|B_s'|) \land \\ & \varphi_3(S_1',\ldots,S_n') \land \bigwedge_{i = 1}^n S_i' = \bigcup_{\beta \models S_i} S'^{\beta} \land \bigwedge_{i = 1}^n S_i = \bigcup_{\beta \models S_i} S^{\beta} \land \\ & \bigwedge_{\beta = 1}^{2^n} S^{\beta} = \{ i \mid i \in S_1^{\beta(1)} \land \ldots \land i \in S_{l-1}^{\beta(l-1)} \land  \exists v. \varphi^{\beta(l..n)}(v,a_2[i],\ldots,a_t[i], \overline{c}) \} \land  \\
& \bigwedge_{\beta = 1}^{2^n} S'^{\beta} \subseteq S^{\beta} \land\bigwedge_{\beta_1 = 1}^{2^n} \bigwedge_{\beta_2 = 1, \beta_1 \neq \beta_2}^{2^n} S'^{\beta_1} \cap S'^{\beta_2} = \emptyset \land S_1 \cup \ldots \cup S_n = S_1' \cup \ldots \cup S_n'
\end{split}
\end{align}
which has eliminated the existential quantifier on the array variable $a$.

Let us then show that Formulas~(\ref{eq:qelimfv}) and (\ref{eq:afterqefv}) are equivalent. 

Indeed, assume that Formula~(\ref{eq:qelimfv}) is true. Then there exists an array $a$ such that the formula holds. Now, let 
\[
S'^{\beta} := \{ i \mid i \in S_1^{\beta(1)} \land \ldots \land i \in S_{l-1}^{\beta(l-1)} \land \varphi^{\beta}(a[i],a_2[i],\ldots,a_t[i],\overline{c}) \} 
\]
It follows that $S_j'$ partitions the universe as required by the second formula. Moreover, $S'^{\beta} \subseteq S^{\beta}$. Finally, the rest of the formula is satisfied by hypothesis and the definition of $S'^{\beta}$. 

Conversely, let us show that if Formula~(\ref{eq:afterqefv}) holds then Formula~(\ref{eq:qelimfv}) holds too. Given any index $i$, it has to belong to some region $S'^{\beta}$ since $S'^{\beta}$ forms a partition of $I$. Since $S'^{\beta} \subseteq S^{\beta}$, it follows that if we define $a[i] = v$ satisfying the definition of $S^{\beta}$ then $S'^{\beta} \subseteq S^{\beta}$ (where $S^{\beta}$ is defined as in (\ref{eq:qelimfv})). Since $S^{\beta}$ (as defined in (\ref{eq:qelimfv})) and $S'^{\beta}$ partition $I$, it follows that $S'^{\beta} = S^{\beta}$ and thus, the rest of the formula in (\ref{eq:qelimfv}) is also satisfied. 

Formula~(\ref{eq:afterqefv}) is equisatisfiable with the following:
\begin{align}
\begin{split}
\exists &S'^{\beta_1},\ldots,S'^{\beta_{2^n}}.  \varphi_1(i_1,\ldots,i_n, B_1',\ldots,B_s') \land \varphi_2(k_1,\ldots,k_m,|B_1'|,\ldots,|B_s'|) \land \\ & \varphi_3(S_1',\ldots,S_n') \land \bigwedge_{i = 1}^n S_i' = \bigcup_{\beta \models S_i} S'^{\beta} \land \bigwedge_{i = 1}^n S_i = \bigcup_{\beta \models S_i} S^{\beta} \land \\ & \bigwedge_{j = l}^{n} S_{j} = \{ i \mid \exists v. \varphi_j(v,a_2[i],\ldots,a_t[i], \overline{c}) \} \land  \\
& \bigwedge_{\beta = 1}^{2^n} S'^{\beta} \subseteq S^{\beta} \land\bigwedge_{\beta_1 = 1}^{2^n} \bigwedge_{\beta_2 = 1, \beta_1 \neq \beta_2}^{2^n} S'^{\beta_1} \cap S'^{\beta_2} = \emptyset \land S_1 \cup \ldots \cup S_n = S_1' \cup \ldots \cup S_n'
\end{split}
\end{align}

To eliminate the remaining existential quantifiers of $S'^{\beta_1},\ldots,S'^{\beta_{2^n}}$, observe that we may eliminate them as in Section~\ref{section:uibapa}. 

Iterating this process, we obtain a formula such as (\ref{eq:original}), but with existential quantifiers in the set comprehensions. 
\begin{align}
\label{eq:eliminated}
\begin{split}
& \varphi_1(i_1,\ldots,i_n, B_1,\ldots,B_s) \land \varphi_2(k_1,\ldots,k_m,|B_1|,\ldots,|B_s|) \land \\ & \varphi_3(S_1,\ldots,S_n) \land \bigwedge_{j = l}^{n} S_j = \{ i \mid \exists v_1,\ldots,v_t. \varphi_j(v_1,\ldots,v_t, \overline{c}) \}
\end{split}
\end{align}

If $\exists v_1,\ldots,v_t. \varphi_j(v_1,\ldots,v_t, \overline{c})$ is true then $S_j$ is the universal set. Otherwise, it is the empty set. Since the theory of Figure~\ref{fig:alberti-syntax} does not contain syntax to refer to the universe set we have to eliminate the occurrences of the universal set using the absolute complement of the empty set.

\begin{cor}
The simple flat array fragment is uniform interpolating.
\end{cor}
\begin{proof}
The proof follows from the quantifier-elimination procedure above and Proposition~\ref{prop:unifqe}.
\end{proof}

\begin{exa}
Let us show how to eliminate existential quantifiers in the expression $\exists x. |\{ i \mid x[i] = 0 \} | = 1$. In a first step, the existential quantifier on the array variable $x$ is transformed into an existential quantifier on an element variable as $|\{ i \mid \exists v. v = 0 \} | = 1$. Now the interpreted set is either the universal set or the empty set. Using the quantifier-elimination procedure for Presburger arithmetic we get that $\exists v. v = 0$ is equivalent to $\top$. Thus, the original formula is equivalent to the quantifier-free formula $|\emptyset^c | = 1$. This can be written as $|\emptyset^c | = 1$.
\end{exa}

\section{Conclusion}
\label{section:conclusion}
We have investigated the quantifier-free interpolation property for several theories of data structures introduced after the original paper of Kapur, Majumdar and Zarba \cite{kapur_interpolation_2006}. 

We have proved that the pointwise generalisation of relations introduced in the theory of combinatory array logic  preserves the (general) quantifier-free interpolation property of the quantifier-free array theory with diffs. Unlike in \cite{bruttomesso_quantifier-free_2012}, it is necessary to introduce iterated diff functions that are furthermore parameterized over the relation symbol used. Profiting from recent advances in \cite{ghilardi_interpolation_2024}, we were also able to prove that this theory does not admit uniform interpolants. 

In order to compensate for the lack of uniform interpolants, we also investigated uniform interpolation in data structure theories that use cardinality constraints. We investigated a theory of sets with cardinalities and set comprehensions that has been repeatedly used in the verification of distributed protocols \cite{dragoi_logic-based_2014, dragoi_psync_2016, gleissenthall_cardinalities_2016, alberti_cardinality_2017, damian_communication-closed_2019}. Reusing well-known quantifier elimination results \cite{feferman_first_1959}, we were able to establish that this theory has uniform interpolants, thanks to the existence of existential quantifier elimination procedures. The problem of whether combinatory array logic or the simple flat array fragment admitted quantifier-free interpolation was open in \cite{ghilardi_interpolation_2023}.

Our results should be contrasted with other (negative) results such as those in \cite{hoenicke_interpolation_2019} which show that array logics that are strongly related to our fragments do not admit the quantifier-free interpolation property. It would also be interesting to explore the application of these interpolation algorithms in the verification of distributed protocols such as those in \cite{berkovits_verification_2019}, which explicitly mentions the computation of interpolants between first-order and cardinality properties.

\bibliographystyle{alphaurl} 
\bibliography{references}

\end{document}